\documentclass[a4paper,11pt,oneside]{article}
\usepackage[english]{babel}
\usepackage[utf8]{inputenc}
\usepackage[T1]{fontenc}
\usepackage{amssymb,amsmath,amsthm,bm,bbold} 
\usepackage{enumerate}
\usepackage{latexsym}
\usepackage{mathtools}
\usepackage{mathrsfs}
\usepackage{fancyhdr}
\usepackage{setspace}
\usepackage{color} 
\usepackage{graphicx} 
\usepackage{subfigure}
\usepackage{caption}
\usepackage{dirtytalk}
\usepackage{stmaryrd}
\usepackage{geometry}
\usepackage[pdfpagelabels,plainpages=false]{hyperref}
\hypersetup{colorlinks=false}
\usepackage{bm}
\usepackage{upgreek}

\numberwithin{equation}{section}

\DeclarePairedDelimiter{\norma}{\lVert}{\rVert}

\DeclareMathOperator{\sgn}{sgn}

\newcommand{\R}{ \mathbb{R}}

\newcommand{\D}{\mathcal{D}}

\newcommand{\n}{\noindent}

\newcommand{\vs}{\vspace{0.5cm}}
\newcommand{\vsa}{\vspace{0.2cm}}

\newcommand{\ba}{\begin{eqnarray}} 
\newcommand{\ea}{\end{eqnarray}}

\newcommand{\be}{\begin{equation}}
\newcommand{\ee}{\end{equation}}

\def\unit{\mathbb{1}}  
\def\io{\infty}
\newcommand{\di}{\textrm{d}}

\let\a=\alpha \let\b=\beta         \let\d=\delta     \let\e=\varepsilon
        \let\k=\kappa     \let\l=\lambda
\let\m=\mu                          \let\r=\rho
          \let\ph=\varphi   
\let\ps=\psi        
\let\G=\Gamma \let\D=\Delta

\def\bR{\mathbb{R}}

{\left\lbrace\begin{array}{@{}l@{}}}%
{\end{array}\right.}

\theoremstyle{plain} 
\newtheorem{lemma}{Lemma}
\newtheorem{theorem}{Theorem}
\newtheorem{proposition}{Proposition}

\theoremstyle{remark}
\newtheorem{remark}{Remark}
 
\theoremstyle{definition}

\title{Universal low-energy behavior in a quantum Lorentz gas with Gross-Pitaevskii potentials}

\author{Giulia Basti$^1$, Serena Cenatiempo$^2$, Alessandro Teta$^1$ \\
\\
Dipartimento di Matematica G. Castelnuovo, Sapienza Universit\`a di Roma,\\ 
Piazzale  Aldo Moro, 5  00185 Roma, Italy$^1$ \\
\\
Gran Sasso Science Institute,\\
Viale Francesco Crispi, 7  67100 L'Aquila, Italy$^2$}

\date{October 27, 2017}

\makeatletter
\let\orgdescriptionlabel\descriptionlabel
\renewcommand*{\descriptionlabel}[1]{%
  \let\orglabel\label
  \let\label\@gobble
  \phantomsection
  \edef\@currentlabel{#1}%
  \let\label\orglabel
  \orgdescriptionlabel{#1}%
}
\makeatother

\begin{document}
\maketitle
\begin{abstract} We consider a quantum particle interacting with $N$ obstacles, whose positions are independently chosen according to a given probability density,  through a two-body potential of the form $N^2 V(Nx)$ (Gross-Pitaevskii potential). We show convergence of the $N$ dependent one-particle Hamiltonian to a limiting Hamiltonian where the quantum particle experiences an effective potential depending only on the scattering length of the unscaled potential and the density of the obstacles. In this sense our Lorentz gas model exhibits a universal behavior for $N$ large. Moreover we explicitely characterize the fluctuations around the limit operator.
Our model can be considered as a simplified model for scattering of slow neutrons from condensed matter.

\end{abstract}

\section{Introduction and main result}

In this note we study the effective behavior at low energy of a non relativistic quantum particle in $\bR^3$ interacting with a system of $N$ randomly distributed obstacles in the limit $N \rightarrow \infty$. In order to formulate such Lorentz gas model, we introduce the set  $ Y_N=\{y_1,\dots y_N\}$ of random variables in $\bR^3$, independenltly chosen according to a common distribution with density $W$.  
We assume that the  interaction among the particle and the $i$-th obstacle is described by the Gross-Pitaevskii potential   
\[  
V^N_i(x) = N^2 V(N(x-y_i))\,,
\]
where the unscaled  potential $V$ decays to zero at infinity sufficiently fast. Therefore, the Hamiltonian of the particle is 
\begin{equation}\label{eq:H_n}
	H_N=-\Delta+\sum_{i=1}^N V_i^N(x)\,,
\end{equation}
where we have chosen units such that $\hbar=1$ and the mass is $1/2$. The assumptions on $V$ will guarantee that $H_N$ is a selfadjoint operator in 
 $L^2(\bR^3)$. The aim of this paper is to characterize the limit behavior of $H_N$ and the fluctuations around the limit.
 
\n
 We note that for $N$ large the range $r_0$ of the potential $V_i^N$ is of order $N^{-1}$  while the average distance $d$ among the obstacles is of order $N^{-1/3}$. If the wavelength  of the particle $\lambda_p$ is taken of order $1$, we are studying the regime 
\[
r_0 \ll d \ll \lambda_p \,,
\]
which is the case occurring, for example, in the analysis of scattering  of slow neutrons from condensed matter (Neutron Optics). We reasonably expect that, for $N \rightarrow \infty$, the particle \say{sees} an effective potential depending on the density of obstacles. Moreover, one could be tempted to consider essentially correct the formal manipulation
\[
\sum_{i=1}^N V_i^N (x) \;\sim \; \frac{1}{N} \sum_{i=1}^N N^3 V(N(x-y_i))\;\sim\; b \, \frac{1}{N} \sum_{i=1}^N \delta(x-y_i)\,,\;\;\;\;\;\;\;\; b=\int\!\!dx\, V(x)\,
\] 
and to obtain $b W$ as effective potential. Indeed, this is not the case and we shall see that the effective potential is the density of scattering length of the system of obstacles $4 \pi aW$, where $a$ is the scattering length associated to the potential $V$ (see definition below). 
The situation is completely analogous to the more difficult case of a gas of $n$ particles interacting via two-body potentials  scaling as $n^2 V(n x)$ for $n \rightarrow \infty$ as investigated in~\cite{ESY0, ESY2, ESY3, P, BDS, BS}. In particular, we refer to \cite[Sect. 5]{BPS} for a discussion on the emergence of the scattering length in that context.

Let us introduce the definition of scattering length. Given the solution  $\phi_0$ of  the zero energy scattering problem
	\begin{equation}\label{eq:phi0}
		\left\{
			\begin{aligned}
				(-\Delta+V)\phi_0=0\\
				\lim_{|x|\to+\infty}\phi_0(x)=1 \,,
			\end{aligned}
		\right.
	\end{equation}
 the scattering length $a$ associated to the potential $V$ is defined by
	\[  
		a=\frac{1}{4\pi}\int dx\,V(x)\phi_0(x). 
\]
It is well known that a condition  for the existence of a finite scattering length is the fact that zero is not an eigenvalue nor a resonance for $-\Delta +V$.  As for the physical meaning, we recall that $a$ represents the effective linear dimension of the scatterer at low energy. It is also  easy to check by scaling  that the scattering length associated to the rescaled potential $V^N_i$ is $\displaystyle{a_i^N=a/N}$. 

\n
In this paper we give the proof of the  convergence in  the strong resolvent sense of $H_N$ to the limiting Hamiltonian
\[ 
H=-\D+ 4\pi a W\,,
\]
where the convergence is in probability with respect to the distribution of the obstacles. Denoted by $\|\cdot\|_p$ the norm in $L^p(\bR^3)$ $1\leq p\leq \infty$ we give below the precise formulation  of our main theorem.

\begin{theorem}\label{MainTheorem} 
 Let $V \in L^1(\bR^3, (1 + |x|^4) \di x)\cap L^3(\bR^3)$ such that zero is not an eigenvalue nor a resonance for $-\Delta +V$ and let $a \in \bR$ be the corresponding scattering length. Moreover, let $W \in L^1(\bR^3) \cap L^p(\bR^3)$, for some $p>3$,  $f \in L^2(\R^{3})$ and take $\l>0$ large enough.
 Then for any $\e>0$ and  $\b<1/2$ we have
\[
	\lim_{N\to +\infty} P_N\big(\,\{Y_{N}: N^{\beta} \| (H_{N}+\lambda)^{-1}f-(H+\lambda)^{-1}f\|_2>\epsilon\} \,\big)=0 \,,
\]
where $P_N$ is the product probability measure $\{W(x) \di x \}^{\otimes N}$ on the set of configurations of points $Y_N$.
\end{theorem}
\begin{remark} Theorem \ref{MainTheorem} implies the convergence in probability as $N \to \io$ of the unitary group $e^{-i t H_N}$, associated to the $N$ dependent Hamiltonian \eqref{eq:H_n}, to the $N$ independent unitary group  $e^{-i t H}$, for any time $t>0$. 
\end{remark}

\n
As an immediate consequence of Theorem \ref{MainTheorem} and of  previous results \cite{FOT, FHT} we can also characterize the fluctuations around the limit operator,  as expressed in the following theorem.

\begin{theorem}\label{th:fluctuations}   
Under the same assumptions of Theorem \ref{MainTheorem}, for any $f,g \in L^2(\bR^3)$ the random variable 
\[
\eta^\l_{f,g}(Y_N) := \sqrt N \,\Big( g, \big( (H_N +\l)^{-1} - (H + \l)^{-1}\big) f \Big)
\]
converges in distribution for $N \to \infty$ to a Gaussian random variable  $\bar \eta^\l_{f,g}$ of zero mean and covariance
\[ \begin{split}
& E\Big( (\bar \eta^\l_{f,g})^2 \Big) \\
& = (4 \pi a)^2 \| (H +\l)^{-1} g\, (H+\l)^{-1}f\|^2_{L^2_W} - 4 \pi a \big( (H+\l)^{-1} g, (H+\l)^{-1} f \big)^2_{L^2_W}
\end{split}\]
where $E(\cdot)$ means expectation with respect to the probability measure $P_N$ and $L^2_W = L^2(\bR^{3}, W(x) \di x )$. 
\end{theorem}

\n
Let us briefly comment on the above results.  We find that the asymptotic behavior  of our Lorentz gas is completely characterized by the density of the obstacles and by  their  scattering length. In particular, this means that  the dependence of the limit Hamiltonian on the interaction potentials $V_i^N$ is only through the associated scattering length $a/N$, i.e., a single physical parameter describing the effect of the obstacle as a scatterer at low energy. In this sense, in our scaling and for $N$ large, the Lorentz gas exhibits a universal behavior. 

 As we already mentioned, in the many-body context the same type of universal behaviour of the interaction arises in the effective description of the dynamics of $n$ bosons interacting through Gross-Pitaevskii potentials and undergoing Bose-Einstein condensation. More precisely, under the assumption that at time zero the system exhibits Bose-Einstein condensation into the one-particle wave function $\varphi \in L^2(\bR^3)$, one expects condensation to be preserved at any time in the limit $n \to \infty$ and the condensed wave function to evolve according to the Gross-Pitaevskii equation $i \partial_t \varphi_t = -\D \varphi_t + 4 \pi a |\varphi_t|^2 \varphi_t$, with initial condition $\varphi_0=\varphi$.  This fact has been well established mathematically for non negative potentials~(see \cite{ESY0, ESY2, ESY3, P, BDS, BS}) and shows that at the level of the evolution of the condensate wave function and in the limit $n \to \infty $ the interaction enters only through its scattering length.

Indeed,  our Lorentz gas can be considered as  a simplified model obtained from the more general case of a test particle interacting with other $N$ particles when the masses of these particles are infinite.  Nevertheless, we believe that our analysis could have some interest and could give some hints for the general case. The reason is that, because of its simpler structure, our Lorentz gas allows a more detailed analysis. In particular, we obtain the convergence result  without any assumption on the sign of the interaction potential $V$ and we can characterize the fluctuations in a relatively simple and explicit way.

\section{Line of the proof}\label{line}

In this Section we describe the method of the proof and collect some preliminary results and notation useful in the sequel. 

\n
Let us start with some notation. Given $ \underline{\phi}_N = \{ \phi_1, \ldots, \phi_N\} \in \oplus_{i=1}^N L^2(\bR^3)$,  we define
 \[
\|\underline{\phi}_N\|^2 = \sum_{i=1}^N \| \phi_i \|^2_2 \,.
\] 
Moreover, for $ \vec{X}_N =\{X_1, \ldots, X_N\} \in \bR^N$  we set 
\[
 \|  \vec{X}_N\|^2 = \sum_{i=1}^N X_i^2\,. 
\]

\n
It is useful to write the interaction potential as  $V(x)=u(x)v(x)$, where
\[ 
		u(x)=|V(x)|^{1/2} \,, \;\;\;\;\;\;\;\;\;\;\; 
		v(x)=|V(x)|^{1/2}\sgn(V(x))\,.
\]
Using the above  factorization, we  rewrite the scattering length associated to the potential $V$ as 
\begin{equation}\label{scalen2}
		 a= \frac{1}{4\pi} (u,\mu)\,, 
	\end{equation}
	where $\mu$ solves
	\begin{equation}\label{eq:mu}
		\mu+v\mathcal{G}^0u\mu=v
	\end{equation}
	and $\mathcal{G}^0$ is the operator with integral kernel $\mathcal{G}^0(x)=(4\pi |x|)^{-1}.$
	Indeed, under the assumption that zero is neither an eigenvalue nor a resonance for $-\Delta+V$, the equation \eqref{eq:mu} has a unique solution in $L^2(\bR^3)$. Then, one can check that the function $\phi_0 := 1-\mathcal{G}^0 u \mu$ solves problem \eqref{eq:phi0} 
\[
(-\Delta +V)(1-\mathcal{G}^0 u\mu)= -u\mu +V -V\mathcal{G}^0u\mu= -u\,(\mu + v \mathcal{G}^0 u \mu - v )=0\,.
\] 
Moreover
\[
4 \pi a = \int \!\!dx\, V \phi_0 = \int\!
\! dx\, u(v-v\mathcal{G}^0u \mu)= \int \!\! dx\, u \mu \,,
\] 
so that \eqref{scalen2} is verified.

\n
Analogously,  for the rescaled potentials we set  $V_i^N(x)=u_i^N(x)v_i^N(x)$, where
\begin{equation*} 
	\begin{aligned}
		u_i^N(x)&=|V_i^N(x)|^{1/2}=Nu(N(x-y_i))\,, \\
		v_i^N(x)&=|V_i^N(x)|^{1/2}\sgn(V_i^N(x))=Nv(N(x-y_i))
	\end{aligned}
\end{equation*}
and for the scattering length we have

	\begin{equation*} 
		 a_i^N= \frac{1}{4 \pi} (u_i^N,\mu_i^N)=a /N \,,
	\end{equation*}
	where
	\be\label{luscan}
		\mu_i^N+u_i^N \mathcal{G}^0 v_i^N \mu_i^N=v_i^N\,.
	\ee

\vskip 0.3cm

\n
Let us discuss the line of the proof. We first observe that 
the proof of Theorem  \ref{MainTheorem} is non probabilistic. In fact,  we prove the convergence for a fixed  set of configurations of obstacles $Y_N=\{y_1,\dots y_N\}$  satisfying the following regularity assumptions
\begin{description}
	\item[(Y1)\label{ass:Y1}]  Let $\nu^*(p)= \frac 1 3 \frac{p-3}{p-1}  \in (0,1/3)$. For any $0<\nu < \nu^*(p)$  there exists $C$ such that
	\[
			\min_{i\neq j}|y_i-y_j|\geq \frac{C}{N^{1-\nu}}\,.
		\]
	\item[(Y2)\label{ass:Y2}]  For any $N>0$ and any $0<\xi\leq 1$ we have
		\[   
			\frac{1}{N^2}\sum_{\substack{i,j=1\\i\neq j}}^N\frac{1}{|y_i-y_j|^{3-\xi}}\leq C_\xi<\infty \,.
		\]
\end{description}
The convergence in probability then follows once we show that  \ref{ass:Y1} and \ref{ass:Y2} hold with probability increasing to one in the limit $N \to \infty$. More precisely the following lemma holds.

\begin{lemma}\label{lm:Y}  Let $Y_N=\{y_1,\dots y_N\}$ a configuration of $N$ identically distributed random variables in $\bR^3$, whose distribution has density $W \in L^1(\bR^3) \cap L^p(\bR^3)$, for some $p>3$. Then, the set of configurations  on which \ref{ass:Y1} and \ref{ass:Y2} hold has probability increasing to one as $N$ goes to infinity.
\end{lemma}

\begin{proof} The standard proof of the Lemma~\cite{O, PV} can be easily adapted to the situation where  $W \in L^1(\bR^3) \cap L^p(\bR^3)$ with $p>3$. 
We start   analysing the assumption \ref{ass:Y1}. Let $Z_N =\{ Y_N | \min_{i\neq j}|y_i -  y_j| \geq   \frac C {N^{1-\nu}}\}$ the set of configurations of $N$ obstacles for which \ref{ass:Y1} holds. We show that in the limit $N \to \io$ the probability of the complement of  $Z_N$ goes to zero. We have
\be \begin{split} \label{PGamma}
P_N(Z_N^c)  =\; & P_N \big( \{   Y_N |\, \exists\,  i,j;i\neq j : |y_i -  y_j| <  C N^{-(1-\nu)} \} \big) \\
\leq & \; \frac{N(N-1)}{2} \int_{|x| <C N^{-(1-\nu)}} \di x\, W(x)\,.
\end{split}\ee
To bound the last integral we use H\"older inequality. For $1/p^\prime+1/p =1$, we have
\[ \begin{split}
 \int_{|x| < CN^{\nu-1}} \hskip -1cm \di x\, W(x) & \leq   \left(\int_{|x| < CN^{\nu-1}} \di x\, \right)^{1/{p^\prime}} \!\!\left(\int_{|x| < CN^{\nu-1}} \di x |W(x)|^p\, \right)^{1/{p}} \\[6pt]
 & \leq CN^{3(\nu-1)/p^\prime} \|W\|_p\,.
\end{split}\]
Hence the r.h.s. of  \eqref{PGamma} goes to zero as $N$ goes to infinity for any $\nu < \frac 1 3 \frac{p-3}{p-1}$. Note that the requirement $W \in L^p$ with $p>3$ assures that $\nu>0$.  \\[-9pt]

To show that also \ref{ass:Y2} holds with probability increasing to one as $N \to \infty$ it is sufficient to note that the $M=N(N-1)$ random variables  $|y_i - y_j|$, with $i,j=1, \ldots, N$  and $i \neq j$, are interchangeable and we can reorder them as $\{ X_1, \ldots, X_M\}$ (e.g. using a diagonal progression as in the Cantor pairing function). Standard results~\cite{K} ensure that under the assumption
\be \label{EX1}
 E( X_{1}^{-3+\xi}) = \int \di x \di y \frac{W(x) W(y)}{|x-y|^{3-\xi}} \leq C\,,
\ee
we have 
\[
\lim_{M \to \io} \frac 1 M \sum_{k=1}^M \frac 1 {X_k^{3-\xi}} = E( X_1^{-3+\xi})\,.
\]
The bound \eqref{EX1} follows from the assumptions on $W(x)$. Let $R>0$ be arbitrary. Then:
\[
	\int_{|x-y|>R} \di x \di y \frac{W(x)W(y)}{|x-y|^{3-\xi}} \leq R^{-3+\xi} \| W\|_1^2 \leq C_\xi\,.
\]
On the other hand
\[
\int_{|x-y| \leq R} \di x \di y \frac{W(x)W(y)}{|x-y|^{3-\xi}} \leq  \int_{|x-y| \leq R} \di x \di y \frac{|W(x)|^2}{|x-y|^{3-\xi}} \leq C_\xi \,.
\]
\end{proof}

\vs

\n
By Lemma \ref{lm:Y} we conclude that, in order to prove Theorem \ref{MainTheorem}, it is enough to show that for all $f\in L^2(\R^{3})$
\be\label{conresy}
	\lim_{N\to\infty}\| (H_{N}+\lambda)^{-1}f-(H+\lambda)^{-1}f\|_2=0
\ee
uniformly on configurations $Y_{N}$ satisfying \ref{ass:Y1} and \ref{ass:Y2}. In fact, since the measure of the configurations where \ref{ass:Y1} and \ref{ass:Y2} do not hold goes to zero as $N \to \infty$, we have
\[ \begin{split}
\lim_{N\to +\infty} &P_N\big(\,\{Y_{N}: N^{\beta} \| (H_{N}+\lambda)^{-1}f-(H+\lambda)^{-1}f\|_2>\epsilon\} \,\big)\, \\
&=  \lim_{N\to +\infty} P_N\big(\,\{Y^*_{N}: N^{\beta} \| (H_{N}+\lambda)^{-1}f-(H+\lambda)^{-1}f\|_2>\epsilon\} \,\big) =0\,,
\end{split}
\]
where $\{Y^*_N\}$ is the set of configurations of obstacles where \ref{ass:Y1} and \ref{ass:Y2} hold.

\n
The proof of Theorem \ref{th:fluctuations} is obtained with slight modifications of the step followed in \cite{FOT}, \cite{FHT} for a similar problem, and therefore we refer the reader to those papers. \\ 

\n
The following remarks summarize two important consequences of the validity of the assumptions \ref{ass:Y1} and \ref{ass:Y2}.

\begin{remark}	
	Notice that from \ref{ass:Y1} and \ref{ass:Y2} it follows  that for any $0<\nu<\nu^*$
	\be  \label{Y3}
	\frac{1}{N^{3-\nu^2}}\sum_{\substack{i,j =1 \\i\neq j}}^N\frac{1}{|y_i-y_j|^{4}} \leq C_\nu\,.
	\ee
Indeed
	\begin{equation*}
		\begin{aligned}
			\frac{1}{N^{3-\nu^2}}\sum_{\substack{i,j=1\\i\neq j}}^N\frac{1}{|y_i-y_j|^4}&=\frac{1}{N^{1-\nu^2}}\frac{1}{N^2}\sum_{\substack{i,j=1\\i\neq j}}^N \frac{1}{|y_i-y_j|^{3-\nu}}\frac{1}{|y_i-y_j|^{1+\nu}}\\
																																			&\underset{\textup{by \ref{ass:Y1}}}{\leq}\frac{1}{N^{1-\nu^2}} N^{(1+\nu)(1-\nu)}\Bigg(\frac{1}{N^2}\sum_{\substack{i,j=1\\ i\neq j}}^N\frac{1}{|y_i-y_j|^{3-\nu}}\Bigg)\\
																																			&\underset{\textup{by \ref{ass:Y2}}}{\leq} c_\nu \,.
		\end{aligned}
	\end{equation*}
\end{remark}

\vsa
\begin{remark}
	For $\lambda\geq 0$ we denote by $\mathcal{G}^{\lambda}$  the free resolvent $(-\Delta+\lambda)^{-1}$ and, with a slight abuse of notation, also the corresponding integral kernel 
	\[ 
		\mathcal{G}^\lambda(x)=\frac{e^{-\sqrt{\lambda}|x|}}{4\pi |x|}.
	\]
Moreover, we define the $N\times N$ matrix $G^\l$ with entries 
\begin{equation*}
	G^{\lambda}_{ij}=
			\begin{cases}
				\mathcal{G}^\lambda(y_i-y_j)  & i\neq j\\
				 0 & i = j  \,.
			\end{cases}
	\end{equation*}

\n
	Then,  due to hypothesis \ref{ass:Y2} on our configuration of obstacles,  we get 
	\[   
		\frac{1}{N}\left\lVert G^\lambda \right\rVert\leq c(\lambda)\to 0 \qquad \text{for $\lambda\to +\infty$}\,.
	\]
  Indeed, if we  fix $0 <\b<1$ and  use that $e^{-x} \leq x^{-\b}$,  we have by \ref{ass:Y2}
	\[
		 \begin{split}
			\frac 1 {N^2 } \| G^{\lambda} \|^2 & \leq  
			 \frac 1 {N^2} \sum_{\substack{i,j=1\\ i\neq j}}^N \frac{e^{-2 \sqrt \l |y_i-y_j|}}{16 \pi^2 |y_i-y_j|^2} \\
							  & \leq c\,  \l^{-\b/2} \frac{1}{N^2} \sum_{\substack{i,j=1\\ i\neq j}}^N \frac 1 {|y_i - y_j|^{2+\b}} \leq c_\b \l^{-\b/2}\,.
		\end{split}
	\]
	 Hence, in particular, there exists $\lambda_0>0$ such that 
	\be\label{eq:lambda0}
		\frac{1}{N}\lVert G^{\lambda}\rVert<1 \qquad \forall \lambda>\lambda_0.
	\ee
	Note that with a slight abuse of notation we denote with the same symbol $G^\l_{ij}$ both the elements of the matrix $G^\l$ and the operator on $L^2(\R^3)$ acting as the multiplication by $G^\l(y_i - y_j)$.
\end{remark}

\vs

\n
Given a set of configurations satisfying $(Y1), (Y2)$, the  strategy for the proof of \eqref{conresy} is based on some ideas and techniques developed in the study of boundary value problems for the Laplacian on randomly perforated domains, see \cite{FOT}, \cite{FT} and \cite{FT2}. For a given $f \in L^2(\bR^3)$ we consider the solution $\ps_N$ of the equation
\[ 
(H_N +\l) \ps_N = f\,.
\]
We use the Resolvent Identity to rewrite
\be \label{eq:psiN}
	\psi_N=(H_N+\lambda)^{-1}f=\mathcal{G}^\lambda f+  \sum_{i=1}^N  \mathcal{G}^\lambda v_i^N \rho_i^N\,,
\ee
where the functions 
\[
\rho_i^N=u_i^N (H_N+\lambda)^{-1}f
\]
solve
\begin{equation}\label{eq:rho}
	\rho_i^N+u_i^N\mathcal{G}^\lambda v_i^N\rho_i^N+u_i^N\sum_{\substack{j=1\\ j\neq i}}^N \mathcal{G}^\lambda v_j^N\rho_j^N=-u_i^N \mathcal{G}^\lambda f.
\end{equation}
The idea is to  represent the potential on the r.h.s. of~\eqref{eq:psiN} by its multipole expansion and  to show that, for large $N$, the only relevant contribution comes from the first term of this expansion, that is the monopole term. According to this program we decompose $\ps_N = \tilde \ps_N  + (\ps_N - \tilde \ps_N)$, with
\begin{equation}\label{eq:psi_N_tilde}
			\tilde{\psi}_N(x)=(\mathcal{G}^\lambda f)(x)+\sum_{i=1}^N \mathcal{G}^\lambda(x-y_i)Q_i^N
		\end{equation}
where
\be\label{eq:Q_N}
Q_i^N=(v_i^N,\rho_i^N)\,.
\ee 
		The problem is then split in two parts: find a limit of $\tilde \ps_N$ and than show that the difference $ (\ps_N - \tilde \ps_N)$ converges to zero for $N$ going to infinity.   
	
	\n	
		In order to find a limit of $\tilde \ps_N$ we recognize that the equation for the charge $Q_i^N$ can be written as
\be
\begin{split} \label{eq:Q}
\frac N{4 \pi a}\, Q_i^N + \sum_{\substack{j=1\\ j\neq i}}^N G^\l_{ij} Q^N_j &= - (\mathcal{G}^\l f)(y_i) + R^N_i
%
\end{split}
\ee
with  $R^N_i= A^N_i + B^N_i + D^N_i$ and
\begin{equation}\label{eq:R}
\begin{aligned}
A_i^N & = \int \di x\, v_i^N (x) \int \di z \frac{e^{-\l |x-z|}-1}{4\pi|x-z|} u^N_i(z) \mu^N_i(z)   \\
B_i^N & = \sum_{\substack{j=1\\ i\neq j}}^N \int \di x\, u_i^N(x) \m_i^N(x) \int \di z \big(\mathcal{G}^\l(x-z) - \mathcal{G}^\l(y_i-y_j) \big) v_j^N(z)  \r^N_j(z)  \\
D_i^N & = \int \di x\, \m_i^N(x) u^N_i(x) \int \di z\, \big(\mathcal{G}^\l(x-z) - \mathcal{G}^\l (y_i - z) \big) f(z)\,.
\end{aligned}
\end{equation}
Since we expect $R_i^N$ to be an error term,  
equation  \eqref{eq:Q} suggests to study the approximate equation obtained from \eqref{eq:Q}  removing $R_i^N$. With this motivation, we define 
\begin{equation}\label{eq:psi_N_hat}
			\hat{\psi}_N(x)=(\mathcal{G}^\lambda f)(x)+\sum_{i=1}^N \mathcal{G}^\lambda(x-y_i)q_i^N\,.
		\end{equation}
		where the charges $q_i^N$ satisfy 
			\begin{equation}\label{eq:q_i}
			\frac{N}{4\pi a}q_i^N+
			\sum_{\substack{j=1\\ j\neq i}}^N
			G^\lambda_{ij}q_j^N=-(\mathcal{G}^\lambda f)(y_i)\,,
		\end{equation}
		As before we rewrite $\tilde \ps_N = \hat \ps_N + (\tilde\ps_N - \hat\ps_N)$ and prove that the difference $ (\tilde\ps_N - \hat\ps_N)$ converges to zero for $N$ going to infinity.  Finally, we show that the sequence $\hat \ps_N$ converges to $\ps$ defined by 
		\begin{equation}\label{eq:psi}
		\psi=(-\Delta+4 \pi aW+\lambda)^{-1}f\,.
	\end{equation}
	This last step is strongly based on the analogy with the Hamiltonian with $N$ zero-range interactions considered in \cite{FHT}. 
	%
	In fact, the resolvent of an Hamiltonian $H_{N\alpha,Y_{N}}$ with $N$ point interactions located at the points $Y_{N}=\{y_1,\dots,y_N\}$ with strength $N\alpha$  is given (see e.g.~\cite{AGHH}) by
\[
	(H_{N\alpha,Y_{N}}+\lambda)^{-1}=\mathcal{G}^{\lambda}+\sum_{\substack{i,j=1\\ i\neq j}}^N (\,\Upxi_{N\alpha,Y_{N}}(\lambda) )^{-1}_{ij}(\mathcal{G}^{\lambda}(\cdot-y_i),\cdot)\mathcal{G}^{\lambda}(\cdot-y_j)\,.
\]
where the $N\times N$ matrix $\Upxi_{N,Y_N}(\lambda)$ is defined by
\[
	(\,\Upxi_{N\alpha,Y_{N}}(\lambda))_{ij}=\left(N\alpha+\frac{\sqrt{\lambda}}{4\pi}\right)\delta_{ij}-(1-\delta_{ij})G^{\lambda}_{ij}.
\]
Hence
\be\label{eq:phiN}
	\phi_{N}=(H_{N\alpha,Y_{N}}+\lambda)^{-1}f=(\mathcal{G}^{\lambda}f)(x)+\sum_{i=1}^N \mathcal{G}^{\lambda}(x-y_i)\tilde{q}_i
\ee
where
\be\label{eq:qtilde}
	\left(N\alpha+\frac{\sqrt{\lambda}}{4\pi}\right)\tilde{q}_i-\sum_{\substack{j=1\\ j\neq i}}^NG^\lambda_{ij}\,\tilde{q}_j=(\mathcal{G}^\lambda f)(y_i).
\ee
Comparing \eqref{eq:phiN} and \eqref{eq:qtilde} with \eqref{eq:psi_N_hat} and \eqref{eq:q_i} respectively and recalling in addition that the scattering length of a point interaction with strength $\alpha$ equals $-1/(4\pi\alpha)$, the analogy between $\phi_N$ and $\hat{\psi}_N$ is clear.  \\
	
	\n
	The paper is organized as follows. In Section \ref{prelBounds} we collect some properties of $\mu_i^N$ and $\rho_i^N$ that will be used along the paper. In Section \ref{sec:monopole}  and \ref{sec:pointcharge} we show that the differences $ (\ps_N - \tilde\ps_N)$ and $ (\tilde\ps_N - \hat\ps_N)$ converge to zero for $N$ going to infinity. In Section \ref{sec:convergence} we study the convergence of $\hat \ps_N$. 
To not overwhelm the notation, from now on we skip  the dependence on $N$ where not strictly necessary. 

\section{A priori estimates} \label{prelBounds}

We prove some useful a priori estimates for the solutions of equations \eqref{luscan} and \eqref{eq:rho}.

\begin{lemma} \label{lemma:rho-mu}   
Let  $\mu_i= \mu_i^N \in L^2(\bR^3)$ and $\rho_i= \r_i^N \in L^2(\bR^3)$ defined in \eqref{luscan} and \eqref{eq:rho} respectively. Under the assumptions of Theorem~\ref{MainTheorem} there exists a constant $C>0$ 
such that
\begin{align}
 \sup_i \| \m_i \|_2 & \leq  C N^{-1/2} 
\end{align}
Moreover
\be
 \| \underline{\mu} \|  \leq C \,, \qquad \| \underline{\r}\,\|    \leq C \| f \|_2 \,.   \label{rho-norm} 
\ee
\end{lemma}

\begin{proof}
It is simple to check that by scaling $\mu_i (x)=N\mu(N(x-y_i))$. On the other hand, since $\mu$ satisfies \eqref{eq:mu}
and zero is not an eigenvalue nor a resonance for $V$ we can invert the operator $(\unit + u\, \mathcal{G}^0 v)$ and get
\be \label{normaHatMu}
\| \mu \|_2^2 \leq C\,.
\ee
Hence
\[
	\|\mu_i \|_2^2=\int \di x |\mu_i (x)|^2=\frac{1}{N} \int \di x |\mu(x)|^2=\frac{1}{N}\|\mu\|_2^2\,,
\]
which leads to $ \sup_i \| \m_i \|_2 \leq  C N^{-1/2} $ and $\|\underline{\mu}\|\leq C$. 

Next we prove the bound for $\| \underline \r \,\|$ where we recall that the charge $\r_i$ solves \eqref{eq:rho}.
We set
\[
\hat \r_i(x) :=  N^{-1} \r_i \left(y_i  + x/N\right)\,.
\]
From \eqref{eq:rho} we have
\begin{multline} \label{eq:hatrhoevol}
\big(\unit  + u\, \mathcal{G}^0 v \big) \hat \r_i(x) + (u (\mathcal{G}^{\l/N^2} - \mathcal{G}^0) v \hat \r_i)(x) + \frac 1 N \sum_{\substack{i, j \\ i\neq j}} G^\l_{ij}(u\, v \hat\r_j)(x)\\[3pt]
+ \frac 1 N \sum_{\substack{i, j \\ i\neq j}} \big(u \big(\mathcal{G}^{\l,N}_{ij}- G^\l_{ij} \big) v \hat\r_j\big)(x) = - u(x)\, (\mathcal{G}^\l f)\left( y_i  + x/N \right)\,,
\end{multline}
where $\mathcal{G}^{\l,N}_{ij}$ denotes the operator in $L^2(\R^3)$ with integral kernel
\be \label{GlambdaN}
 \mathcal{G}_{ij}^{\l, N}(x-z) =\frac{e^{-\sqrt \l |y_i -y_j - (x-z)/N|}}{4 \pi^2 |y_i -y_j - (x-z)/N|} \,.
\ee
Our goal is to show that the operator $M^\l$ acting on $\oplus_{i=1}^N L^2(\R^3)$ defined by 
\[ \begin{split}
\big( M^\l \big)_{ij} =\; & \big[(\unit + u\, \mathcal{G}^0 v) +  u (\mathcal{G}^{\l/N^2} - \mathcal{G}^0) v\big] \d_{ij} \\
&+ \frac 1 N  \big[  u G^\l_{ij} v+ u \big(\mathcal{G}^{\l,N}_{ij}- G^\l_{ij} \big) v \big] (1 -\d_{ij})
\end{split}
\]
is invertible. Due to the assumptions on $V$ the operator $(\unit + u \mathcal{G}^0 v)$ is invertible. Then, in order to prove that $M^\l$ is invertible 
it suffices to show that there exists $\l_0$ such that the operators $M_{1}^\l$, $M_{2}^\l$ and $M_{3}^\l$ defined by
\[ \begin{split}
 \big( M_{1}^\l \big)_{ij} &= u (\mathcal{G}^{\l/N^2} - \mathcal{G}^0) v \,\d_{ij} \\
 \big( M_{2}^\l \big)_{ij} &= N^{-1}  u\, G^\l_{ij}\, v  \\ 
  \big( M_{3}^\l \big)_{ij} &=   N^{-1} u \big(\mathcal{G}^{\l,N}_{ij}- G^\l_{ij} \big)  v\,(1-\d_{ij}) \,.
\end{split} 
\]
have a norm going to zero as $N\to \io$ for any $\l >\l_0$.
Denoting with $\lVert\cdot\rVert_{\textup{HS}}$ the Hilbert-Schmidt norm in $L^2(\bR^3)$ and using  the definition of $u(x)$ and $v(x)$, we obtain
\be \label{M1norm}
 \|M_{1}^\l\|^2   \leq    \|u (\mathcal{G}^{\l/N^2} - \mathcal{G}^0) v \|_{HS}^2
 \leq C \frac{\l}{N^2} \|V\|_{1}^2\,,
\ee
which is small for any $\l$ in the limit $N \to \io$.  To bound the norm of the  second matrix, we fix $0 < \b<1$. Then, using \ref{ass:Y2}
\be
\begin{split} \label{M2norm}
 \|M_{2}^\l\|^2  \leq \frac 1 {N^2} \sum_{\substack{i, j \\ i \neq j}} \| u\, G^\l_{ij}\, v \|^2_{HS}  
 & =  \frac 1 {N^2}  \sum_{\substack{i, j \\ i \neq j}}  \int \di x \di z |V(x)| \frac{e^{-2 \sqrt{\l } |y_i - y_j|}}{ 16\pi^2|y_i- y_j|^2} |V(z)| \\
& \leq c \| V \|_1 ^2\, \l^{-\b/2} \frac 1 {N^2}  \sum_{\substack{i, j \\ i \neq j}} \frac{1}{|y_i - y_j|^{2+\b}}  \leq c_\b \l^{-\b/2}\,.
\end{split}
\ee
The r.h.s. of \eqref{M2norm} can be made small by choosing $\l$ sufficiently large. 

To bound the third term we use that for any $\xi < 1$ the following bound holds true: 
\be \label{3.12}
\| u\, (\mathcal{G}^{\l, N}_{ij} - G^\l_{ij})\, v \|^2_{HS} \leq  C\left[ \frac{1}{N^2 |y_i-y_j|^4}   + \frac{1}{N^2 |y_i-y_j|^2}  + \frac{1}{N^{1-\xi} |y_i - y_j|^{3-\xi}} \right]
\ee
From \eqref{3.12} and using the assumptions on the charge distribution and \eqref{Y3} we have
\be
\begin{split}  \label{M3norm}
 \|M_{3}^\l\|^2 & \leq   \frac 1 {N^2} \sum_{\substack{i, j=1 \\ i \neq j}}^N \| u\, (\mathcal{G}^{\l, N}_{ij} - G^\l_{ij})\, v \|^2_{HS} 
 \\ &\leq C  \sum_{\substack{i, j=1 \\ i \neq j}}^N \left[ \frac{1}{N^4 |y_i-y_j|^4}   + \frac{1}{N^4 |y_i-y_j|^2}  + \frac{1}{N^{3-\xi} |y_i -y_j|^{3-\xi}}\right] 
 \leq \frac{C} {N^{1-\xi}}\,.
\end{split}
\ee
To prove \eqref{3.12} we define the cutoff function $\chi_{ij}^N(x)$ to be equal to one if \mbox{$|x| \leq\frac  N 2 |y_i -y_j|$} and zero otherwise and write
\be\begin{split}  \label{3.14}
 & \| u (\mathcal{G}^{\l,N}_{ij} -\mathcal{G}^\l_{ij}) v \|_{HS}^2 \\
 & \leq C \int \hskip -0.1cm \di x \di z \chi_{ij}^N(x-z)  |V(x)| |V(z)| \left(\frac{e^{-\sqrt \l | y_i - y_j -(x-z)/N|}}{|y_i - y_j -(x-z)/N|} - \frac{e^{-\sqrt \l |y_i-y_j|}}{|y_i-y_j|}\right)^2\\
 &  + C\int \hskip -0.1cm\di x \di z |V(x)| |V(z)| (1 - \chi_{ij}^N(x-z)) \hskip -0.1cm \left(\frac{e^{-2\sqrt \l | y_i - y_j -(x-z)/N|}}{|y_i - y_j -(x-z)/N|^2} + \frac{e^{-2\sqrt \l |y_i-y_j|}}{|y_i-y_j|^2}\right)
 \end{split} 
 \ee
To bound the term on the second line of \eqref{3.14} we exploit the fact that whenever $\chi_{ij}^N(x)$ is different from zero the difference in the round brackets is small. In particular, we have
\be\begin{split}  \label{Glambda2}
& \int \di x \di z \chi_{ij}^N(x-z)  |V(x)| |V(z)| \left(\frac{e^{-\sqrt \l | y_i - y_j -(x-z)/N|}}{|y_i - y_j -(x-z)/N|} - \frac{e^{-\sqrt \l |y_i-y_j|}}{|y_i-y_j|}\right)^2\\[6pt]
 & \leq C \hskip -0.1cm\int  \hskip -0.1cm \di x \di z \chi_{ij}^N(x-z) |V(x)| |V(z)| \frac{1}{|y_i-y_j|^2} \left(e^{-\sqrt \l | y_i - y_j -(x-z)/N|}- e^{-\sqrt \l |y_i-y_j|}  \right)^2  \\[6pt]
 & \phantom{{}={}} +C \int  \di x \di z  \chi_{ij}^N(x-z) |V(x)| |V(z)| \; e^{-\sqrt \l | y_i - y_j -(x-z)/N|}\\
 & \hskip 4cm \times\left( \frac{1}{| y_i - y_j -(x-z)/N|} - \frac{1}{ |y_i - y_j|}\right)^2  \,.
\end{split}\ee
To bound the first term on the r.h.s. of  \eqref{Glambda2} we use  the bound
\[
\big|e^{-|y- w/N|} - e^{-|y|} \big| \leq C |w|/N
\]
obtaining
\be \label{bound1}
C \int \di x \di z \chi_{ij}^N(x-z)  |V(x)| |V(z)|  \frac {|x-z|^2} {N^2 |y_i - y_j|^2} \leq  \frac C {N^2} \frac 1 { |y_i - y_j|^2}\,.
\ee
To bound the second term on the r.h.s. of  \eqref{Glambda2} we note that  $|y_i - y_j - (x-z)/N|$ $\geq |y_i - y_j| - |x-z|/N$, and moreover on the support of $\chi_{ij}^N(x-z)$ we also have $|y_i - y_j - (x-z)/N$ $\geq |y_i - y_j|/2$. Hence:
\[
\chi_{ij}^N(x-z) \Big(\frac{1}{|y_i - y_j - (x-z)/N|} - \frac 1{|y_i - y_j |}\Big)  \leq  \frac{2|x-z|}{N |y_i - y_j|^2}
\] 
We obtain 
\be \begin{split} \label{bound2}
&\int  \di x \di z  \chi_{ij}^N(x-z) |V(x)| |V(z)| \; e^{-\sqrt \l | y_i - y_j -(x-z)/N|}\\
 & \hskip 4cm \times\left( \frac{1}{| y_i - y_j -(x-z)/N|} - \frac{1}{ |y_i - y_j|}\right)^2  \\[6pt]
& \leq C \int \di x \di z   \chi_{ij}^N(x-z) |V(x)| |V(z)|   \frac {|x-z|^2} {N^2 |y_i - y_j|^4}\\
&  \leq \frac{C}{N^2 |y_i - y_j|^4}\,.
\end{split}\ee
We are left with the bound of the term on the third line of \eqref{3.14}, for which we exploit the fast decaying behaviour of the potential. We start from the term which does not contain the singularity. We fix $\a$ such that $1-\a \in (0,1)$ and we multiply and divide by $|x-z|^\a$; then we use that $|x-z| \geq  C N |y_i - y_j|$ on the support of the integral and that $|x-z|\geq 1$ for $N$ large enough, due to \ref{ass:Y1}:
\be \label{bound3}
\begin{split}
& \int \di x \di z |V(x)| |V(z)| (1 - \chi_{ij}^N(x-z)) \frac{e^{-2\sqrt \l |y_i-y_j|}}{|y_i-y_j|^2} \\
& \quad \leq \frac{C}{ |y_i-y_j|^2}  \int \di x \di z |V(x)| |V(z)| (1 - \chi_{ij}^N(x-z)) \frac{|x-z|^{\a}}{N^\a |y_i-y_j|^\a} \\
& \quad \leq  \frac{C}{N^\a |y_i-y_j|^{2+\a}} \int \di x \di z |V(x)| |V(z)| |x-z| \\
& \quad \leq  \frac{C}{N^\a |y_i-y_j|^{2+\a}}\,.
\end{split}
\ee
To bound the remaing term on the third line of \eqref{3.14} we use a similar idea; we obtain 
\be \label{3.19}
\begin{split}
& \int \di x \di z |V(x)| |V(z)|(1 - \chi_{ij}^N(x-z)) \frac{e^{-2\sqrt \l |y_i-y_j - (x-z)/N|}}{|y_i-y_j - (x-z)/N|^2} \\
& \leq  \frac{C}{(N |y_i-y_j|)^{2+\a}} \int \di x \di z |V(x)| |V(z)| (1 - \chi_{ij}^N(x-z)) \frac{|x-z|^{2+\a}}{|y_i-y_j - (x-z)/N|^2} \\
& \leq  \frac{C}{ N^\a |y_i-y_j|^{2+\a}} \int \di x \di z |V(x)| |V(z)| (1 - \chi_{ij}^N(x-z)) \frac{(|x| +|z|)^{3}}{|N(y_i-y_j) - (x-z)|^2} \\
& \leq  \frac{C}{ N^\a |y_i-y_j|^{2+\a}} \int \di x^\prime \di z^\prime \frac{|V(x^\prime+ N y_i)| |V(z^\prime+ N y_j)| }{|x^\prime- z^\prime|^2}\, (|x^\prime + N y_i |^3 + |z^\prime + N y_j|^3)  
\end{split}
\ee  
where in the last line we used the change of  variables $x^\prime= x- N y_i$ and $z^\prime= z- N y_j$ and we removed the cutoff function.  

To bound the r.h.s. of \eqref{3.19} we use the Hardy-Littlewood-Sobolev inequality: let $f \in L^p(\bR^n)$ and $h \in L^q(\bR^n)$ with $p,q>1$ and let $0<\l<n$ with $1/p +\l/n +1/q=2$, then there exists a constant $C$ independent of $f$ and $h$ such that
\[
 \Big| \int_{\bR^n} f(x)|x-y|^{-\l} h(y) dx dy \Big| \leq C \|f \|_p \| h\|_q\,.
\]
Hence 
 \[
\begin{split}
& \int \di x \di z \frac{|V(x+ N y_i)| |V(z+ N y_j)|  |x + N y_i |^3 }{|x- z|^2} \leq C \| |\cdot|^3V   \|_{3/2} \| V \|_{3/2}\,,
\end{split}
\] 
Then, due to our assumptions on the potential, we obtain
\be \label{bound4}
\begin{split}
& \int \di x \di z |V(x)| |V(z)|(1 - \chi_{ij}^N(x-z)) \frac{e^{-2\sqrt \l |y_i-y_j - (x-z)/N|}}{|y_i-y_j - (x-z)/N|^2} \\
& \leq  \frac{C}{ N^\a |y_i-y_j|^{2+\a}}
\end{split}
\ee
Putting together \eqref{bound1}, \eqref{bound2}, \eqref{bound3} and \eqref{bound4} we prove  \eqref{3.12}.

The bounds \eqref{M1norm}, \eqref{M2norm} and \eqref{M3norm} together with the assumptions on $V$ prove that there exists a $\l_0>0$ such that for any $\l>\l_0$ and $N$ large enough the operator $M^\l$  is invertible.  From Eq.~\eqref{eq:hatrhoevol} we obtain
\[
\begin{split}  
\| \hat \rho_i \|^2_2 &\leq  C \sum_{i=1}^N \int \di x |u(x)|^2 |(\mathcal{G}^\l f)(y_i +x/N)|^2  \\
& \leq  C N \Big( \sup_x |(\mathcal{G}^\l f)(x)|^2 \Big)  \\
& \leq C N \|f \|_2^2 \,.
\end{split}
\]
It follows that
\[
\begin{split}  
\|  \underline \rho\,\|^2 &=\sum_{i=1}^N \int \di x | N \hat \rho_i( N(x-y_i))|^2  = \frac 1 N \sum_{i=1}^N \| \hat \rho_i \|_2^2 \leq C \|f \|_2^2 \,.
\end{split}
\]

\end{proof}

%
%

\section{Monopole expansion} \label{sec:monopole}

In this section we analyse the difference between the solution $\ps_N$ defined by \eqref{eq:psiN} and the approximate solution $\tilde \ps_N$, obtained considering the first term of a multipole expansion for the potential, defined in \eqref{eq:psi_N_tilde}. This is the content of the following proposition. 
\begin{proposition}\label{prop:step1}  Let $\psi_N$  and $\tilde{\psi}_N$ be defined in \eqref{eq:psiN} and \eqref{eq:psi_N_tilde}  respectively. Then, under the assumption of Theorem \ref{MainTheorem}, 
\[
	\lim_{N \to \io} N^{\b}\norma{\psi_N-\tilde{\psi}_N}_2=0 \qquad  \forall \b<1\,.
\]
\end{proposition}

\begin{proof}
Using \eqref{eq:psiN} and \eqref{eq:psi_N_tilde} we write
\[
\ps_N(x) - \tilde \ps_N(x) = \sum_{i=1}^N \int \di z\, v_i(z) \rho_i(z) \big( \mathcal{G}^\l(x-z) - \mathcal{G}^\l(x-y_i) \big) := \sum_{i=1}^N K_i(x)\,.
\] 
We have
\be \begin{split} \label{2.1}
\| \ps_N - \tilde \ps_N \|^2 &\leq 
 \sum_{i=1}^N \int \di x\, K_i^2(x) + \sum_{\substack{i,j=1\\ i \neq j}}^N \int \di x\, K_i(x) K_j(x)\,.
\end{split}\ee
We first bound the diagonal term on the r.h.s. of \eqref{2.1}.
\[ \begin{split} 
 \sum_{i=1}^N \int \di x K_i^2(x)  & = \sum_{i=1}^N \int \di  z \di z^\prime v_i(z) v_i(z^\prime) \rho_i(z) \rho_i(z^\prime)  \\
 & \hskip 1cm \times \int \di x \big( \mathcal{G}^\l(x-z) - \mathcal{G}^\l (x-y_i) \big)\big( \mathcal{G}^\l(x-z^\prime) - \mathcal{G}^\l (x-y_i) \big) 
 \end{split}\]
 Using elliptic coordinates we can explicitly calculate the integral over $x$ of the products of Green's functions on the last line. For instance, let us consider the product $\mathcal{G}^\lambda(x-z)\mathcal{G}^\lambda(x-z^\prime).$ We set $r_1=|x-z|$,\, $r_2=|x-z^\prime|$ and $R=|z-z^\prime|$ and consider the new variables  $\{\mu, \nu, \ph\}$ with $\mu=(r_1+r_2)/R \in [1, +\infty)$, $\nu= (r_1-r_2)/R \in [-1,1]$ and $\ph \in[0, 2\pi)$ the rotation angle with respect to the axis $zz^\prime$.
Then
\[ 
\begin{split} 
\int \di x\, &\mathcal{G}^\lambda (x-z)\mathcal{G}^\lambda(x-z^\prime)\\
&=\int dx\,\frac{e^{-\sqrt{\lambda}(|x-z|+|x-z^\prime|)}}{16\pi^2|x-z||x-z^\prime|}\\
&=\int_1^{+\infty}d\mu\int_{-1}^1d\nu\int_0^{2\pi}d\varphi \frac{R^3}{8}(\mu^2-\nu^2)\frac{e^{-\sqrt{\lambda}(\mu+\nu+\mu-\nu)R/2}}{16\pi^2\frac{R}{2}(\mu+\nu)\frac{R}{2}(\mu-\nu)}\\ 																							
&=\frac{1}{8\pi\sqrt{\lambda}}\,e^{-\sqrt{\lambda}|z-z^\prime|}\,.
\end{split}
\]
Proceeding analogously for the other terms we obtain
 \[\begin{split} 
 & \sum_{i=1}^N \int \di x\, K_i^2(x)   \\
 &  \leq C \sum_{i=1}^N \int \di  z \di z^\prime \,v_i(z) v_i(z^\prime) \rho_i(x) \rho_i(z^\prime)  \Big( e^{-\l|z-z^\prime|} -e^{-\l|z-y_i|}  - e^{-\l|z^\prime-y_i|}  + 1 \Big) \,.
 \end{split}\]
We use the definition of $v_i$ and rescale the integration variables as follows $N(z-y_i) \to z$ and $N(z^\prime-y_i) \to z^\prime$. Hence
  \be \begin{split}  \label{Step1Diag}
  & \sum_{i=1}^N \int \di x K_i^2(x) \\
  &  \leq C N^{-3} \sum_{i=1}^N  \int \di  z \di z^\prime| V(z)|^{1/2} |V(z^\prime)|^{1/2} \hat \rho_i(z) \hat \rho_i(z^\prime)   (|z| + |z^\prime|)  \\
 & \leq C N^{-3} \| (1+|\cdot|^2) V \|_1  \sum_{i=1}^N \|\hat \r_i \|_2^2 \leq C N^{-2}
\end{split}
\ee 
where we recall that $\hat \rho_i(z) = N^{-1} \rho_i(y_i + z/N)$ and $ \|  \underline{\hat{\rho}} \|^2 \leq C N$  from Lemma \ref{lemma:rho-mu}. 

To bound the non diagonal term we use Cauchy-Schwarz  inequality and the elementary estimate $ab\leq 1/2(a^2+b^2)$
    \be \begin{split}  \label{2.3}
   \sum_{\substack{i,j=1\\ i \neq j}}^N \int \di x & K_i(x) K_j(x)  \; \leq   \frac \e 2 \sum_{\substack{i,j=1\\ i \neq j}}^N \| \rho_i \|^2_2 \| \rho_j \|^2_2  \\
   & +\frac 1 {2 \e}  \sum_{\substack{i,j=1\\ i \neq j}}^N \int \di z \di z^\prime |v_i(z)|^2 |v_j(z^\prime)|^2  \\
   &\hskip 1cm \times \Big | \int \di x  \big(\mathcal{G}^\l(x- z)  - \mathcal{G}^\l(x - y_i) \big) \big(\mathcal{G}^\l(x- z^\prime)  - \mathcal{G}^\l(x - y_j) \big) \Big|^2
  \end{split}\ee
with $\e>0$ to be fixed.  To bound the second term in the r.h.s. of Eq.~\eqref{2.3} we first integrate over $x$ using elliptic coordinates, then we use the definition of $v_i$ and rescale the integration variables $z$ and $z^\prime$. We obtain
 \be \label{2.4}
\frac 1 {2 \e N^2}  \sum_{\substack{i,j=1\\ i \neq j}}^N  \int \di z \di z^\prime |V(z)| |V(z^\prime)| |\zeta^N_{ij}(z, z^\prime)|^2\ee
with
\be \label{zetaij}
\zeta^N_{ij}(z, z^\prime)\! =  e^{-\sqrt{\l}|y_i-y_j + (z-z^\prime)/N  |}   -e^{-\sqrt{\l}|y_i-y_j +z/N|} - e^{-\sqrt{\l}|y_i-y_j -z^\prime/N|} +e^{-\sqrt{\l}|y_i-y_j|}\,.
\ee
With a Taylor expansion at first order, it is easy to check that the function $f(x)= e^{-\sqrt{\l} |x+a|}$  satisfies
\[
 \Big|   e^{-\sqrt{\l} |x+a|} - e^{-\sqrt \l |a|} \Big( 1 - \sqrt{\l}\, \frac{x \cdot a}{|a|} \Big)\Big| \leq C |x|^2\,.
\]
with $C$ independent on $a$. Hence
\be \label{differenzaG}
|\zeta^N_{ij}(z, z^\prime)|  \leq \frac{C}{N^2} (|z|+|z^\prime|)^2\,,
\ee
 and
\begin{multline} \label{2.4}
\frac 1 {2 \e N^2} \sum_{\substack{i,j=1\\ i \neq j}}^N  \int \di z \di z^\prime |V(z)| |V(z^\prime)|\,| \zeta^N_{ij}(z, z^\prime)|^2 \\
\leq \frac C {\e  N^4 }  \int \di z \di z^\prime |V(z)| |V(z^\prime)| (|z|+ |z^\prime|)^4  \leq C \e^{-1} N^{-4}  \,. 
\end{multline}
  Using Eq. \eqref{2.3} and \eqref{2.4}, the bound $\| \underline \rho\, \| \leq C$ and the assumptions on the potential, and choosing $\e = N^{-2}$   we obtain
   \be \begin{split}   \label{Step1NonDiag}
   \sum_{\substack{i,j=1\\ i \neq j}}^N \int \di x & K_i(x) K_j(x)  \; \leq  C N^{-2}\,.
  \end{split}\ee
Eq. \eqref{Step1Diag} and \eqref{Step1NonDiag}, together with \eqref{2.1} conclude the proof of the proposition.
\end{proof}
%
%

\section{Point charge approximation}  \label{sec:pointcharge}
In this section we analyse the difference between $\tilde{\psi}_N$ and $\hat{\psi}_N$ and show that it is small for large $N$. This is the content of the next proposition.

\begin{proposition}\label{prop:step2} Let $\tilde{\psi}_N$ be defined by \eqref{eq:psi_N_tilde}, \eqref{eq:Q_N} and $\hat{\psi}_N$ by \eqref{eq:psi_N_hat} and \eqref{eq:q_i}. Under the assumptions of Theorem \ref{MainTheorem}  and for $\lambda$ large enough
\[
 \lim_{N \to \io}N^{-\b}	\norma{\hat{\psi}_N-\tilde{\psi}_N}_2 =0 \qquad \forall \b<3/2
\]
\end{proposition}

The proposition follows from the next two lemmas.
\begin{lemma} \label{lemma:HatVsTilde}
Let $\tilde{\psi}_N$  be defined by \eqref{eq:psi_N_tilde}, \eqref{eq:Q_N} and $\hat{\psi}_N$ by \eqref{eq:psi_N_hat} and \eqref{eq:q_i}. Then,
\[ 
	\norma{\hat{\psi}_N-\tilde{\psi}_N}_2\leq c\, \sqrt{N}\, \norma{ \vec q- \vec Q}
\]
\end{lemma}

\vskip 0.2cm
\begin{lemma}\label{lemma:Qvsq} Let $Q_i=(v_i, \rho_i)$ and $q_i$ be defined in \eqref{eq:q_i}. Then there exists $\d>0$ such that
\[
\| \vec Q - \vec q \,\| \leq  C N^{-2 -\d}\,.
\]
\end{lemma}

\vskip 0.2cm

\begin{proof}[Proof of Lemma \ref{lemma:HatVsTilde}.] We notice that 
\begin{equation} \label{2.6}
	\begin{aligned}
\norma{\tilde{\psi}_N-\hat{\psi}_N}^2_2&\leq 
	\int dx\,\left|\sum_{i=1}^N \mathcal{G}^{\lambda}(x-y_i)(q_i-Q_i)\right|^2\\
 &= \sum_{i=1}^N (q_i-Q_i)^2 \int \di x  \frac{e^{- 2\sqrt \l |x-y_i|}}{16 \pi^2 |x-y_i|^2} \\
 &\phantom{{}={}} + \sum_{\substack{i,j=1\\i\neq j}}^N(q_i-Q_i)(q_j-Q_j)\int dx\, \frac{e^{-\sqrt{\lambda}(|x-y_i|+|x-y_j|)}}{16\pi^2|x-y_i||x-y_j|}.
	\end{aligned}
\end{equation}
The term on the second line of \eqref{2.6} is clearly bounded  by $ C \| \vec Q - \vec q \|^2$.
To evaluate the integral in the last line of \eqref{2.6} we use an explicit integration as in the proof of Prop.~\ref{prop:step1} and Cauchy-Schwarz inequality. We get:
\[
	\begin{aligned}
	 \sum_{\substack{i,j=1\\i\neq j}}^N&(q_i-Q_i)(q_j-Q_j)\int dx\, \frac{e^{-\sqrt{\lambda}(|x-y_i|+|x-y_j|)}}{16\pi^2|x-y_i||x-y_j|} \\
	&\leq c_\lambda\sum_{\substack{i,j=1\\i\neq j}}^N(q_i-Q_i)(q_j-Q_j)e^{-\sqrt{\lambda}|y_i-y_j|}\\
		 &\leq c_\lambda \sum_{\substack{i,j=1\\i\neq j}}^N(q_i-Q_i)^2\\
		 &\leq c_\lambda N\norma{\vec q- \vec Q}^2\,.
	\end{aligned}
\]
\end{proof}

\begin{proof}[Proof of Lemma \ref{lemma:Qvsq}.]  Eqs. \eqref{eq:Q} and \eqref{eq:q_i} for the charges $Q_i$ and $q_i$ give
\be \label{Qminusq}
\frac N{4 \pi a}\,( Q_i - q_i) + \sum_{\substack{j=1\\j \neq i}}^N G^\l_{ij}  (Q_j-q_j) = R_i\,.
\ee
We denote
\be\label{eq:Gamma}
\G^\l_{ij} := \left(\d_{ij} + \frac {4\pi a} N (1-\d_{ij})G^\l_{ij}  \right)\,,
\ee
so that \eqref{Qminusq}  becomes
\[
\sum_{\substack{j=1\\j \neq i}}^N \G^\l_{ij}  (Q_j - q_j) =\frac{4\pi a}{N}R_i\,,
\]
where $R_i=A_i+B_i+D_i$ is defined in \eqref{eq:R}. On the other hand the bound \eqref{eq:lambda0} yields immediately the invertibility of $\G^\lambda_{ij}$ for $\lambda>\lambda_0.$ Therefore 
\[
\| \vec Q-  \vec q\, \| \leq \frac C N \| \vec R\|
\]
  Lemma \ref{lemma:Qvsq} is proved showing that there exists $\d>0$ such that 
\be \label{normR}
\| \vec R\, \|\leq N^{-1 -\d}\,.
\ee
We start from $\| \vec A\,\|$. 
By Cauchy-Schwarz inequality we get
\[
|A_i| \leq  \| \r_i \|_2 \| \m_i \|_2 \| v_i (\mathcal{G}^\l - \mathcal{G}^0) u_i \|_{HS}\,.
\]
Using the definitions of $u_i$ and $v_i$
\[ \begin{split} \label{viDeltaGui}
 \| v_i (\mathcal{G}^\l - \mathcal{G}^0) u_i \|_{HS}^2 & = \int \di x \di z N^4 |V(N(x-y_i))| |V(N(z-y_i))| \frac{(e^{-\sqrt \l |x-z|}-1)^2}{ 16 \pi^2 |x-z|^2} \\
 & = C \int \di x \di z \frac{|V(x)| |V(z)|}{|x-z|^2} \big( e^{-\sqrt{\l} |x-z|/N} -1 \big)^2 \\
 & \leq C  \l N^{-2} \| V\|_1^2 \,.
\end{split}\]
With the bounds in Lemma \ref{lemma:rho-mu} we have
\be \begin{split}  \label{normA}
\| \vec A\, \| &\leq \| \underline \r\, \|  \Big(\sup_i \| v_i (\mathcal{G}^\l - \mathcal{G}^0) u_i \|^2_{HS}\Big)^{1/2} \; \Big(\sup_i \| \m_i \|^2_2 \Big)^{1/2} \leq C  N^{-3/2}\| f\|_2\,.
\end{split}\ee
Next, we analyse $\|\vec B\,\|$. We define 
\[
B_{ij} =\int \di x \di z\, u_i(x) \m_i(x) \big(\mathcal{G}^\l(x-z) - G^\l_{ij} \big) v_j(z)  \r_j(z)   \,.
\]
Then
\[ 
\| \vec B\,\|^2 =\;  \sum_{i=1}^N \Bigg( \sum_{\substack{j=1\\j \neq i}}^N B_{ij} \Bigg)^2 \,.  
\]
Using twice Cauchy-Schwarz inequality, first in the $x$ and $z$ variables and then in the sum over $j$, we get
\[ \begin{split}
\| \vec B\, \|^2 & \leq\sum_{i=1}^N  \Bigg[ \sum_{\substack{j=1\\j \neq i}}^N  \|\mu_i \|_2 \| \r_j \|_2 \Big( \int \di x \di z |u_i(x)|^2 \big(\mathcal{G}^\l(x-z) - G^\l_{ij} \big)^2 |v_j(z)|^2 \Big)^{1/2}  \;\Bigg]^2 \\
& \leq \sum_{i=1}^N  \|\mu_i\|^2_2  \, \sum_{\substack{k=1\\k \neq i}}^N \| \r_k \|^2_2 \, \sum_{\substack{j=1\\j \neq i}}^N \| v_j (\mathcal{G}^\l - G^\l_{ij}) u_i \|_{HS}^2  \\
& \leq \big(\sup_i \|\mu_i\|^2_2  \big) \| \underline \r\, \|^2 \sum_{i=1}^N \sum_{\substack{j=1\\j \neq i}}^N \| v_j (\mathcal{G}^\l - G^\l_{ij}) u_i \|_{HS}^2  \,.
\end{split}
\]
Rescaling variables and recalling the definition \eqref{GlambdaN} for $\mathcal{G}^{\l,N}_{ij}$ we have
\be \label{vjDeltaGui}
 \| v_j (\mathcal{G}^\l - G^\l_{ij}) u_i \|_{HS}^2  = N^{-2}\| v (\mathcal{G}^{\l,N}_{ij} - G^\l_{ij}) u \|_{HS} .
 \ee 
%
Using the bounds \eqref{Y3}, \eqref{Glambda2}  and \eqref{vjDeltaGui}, together with Lemma \ref{lemma:rho-mu}  we obtain
\be \begin{split} \label{normB}
\| \vec B\, \| & \leq C \Bigg(\frac{1}{N^5} \sum_{\substack{i,j =1\\ i \neq j}} ^N \frac{1}{|y_i - y_j|^4} \Bigg)^{1/2} \leq  c_\nu N^{-1 - \frac 1 2 \nu^2}\,.
\end{split}\ee
To conclude we consider $\|\vec D\|$. We have
\[ \begin{split}
\| \vec D \,\|^2  =\; & \sum_{i=1}^N \Big| \int \di x \di z \m_i(x) u_i(x) \, \big(\mathcal{G}^\l(x-z) - \mathcal{G}^\l (y_i - z) \big) f(z)\Big|^2 \\
  =\; & \sum_{i=1}^N \int \di z \di z^\prime f(z) f(z^\prime) \xi_i(z) \xi_i(z^\prime)\,,
\end{split}\]
with
\[
\xi_i(z) = \int \di x \mu_i(x) u_i(x) \big(\mathcal{G}^\l(x-z) - \mathcal{G}^\l (y_i - z) \big)\,.
\]
Using Cauchy-Schwarz inequality 
\be \begin{split}
\| \vec D \,\|^2 
& \leq \| f\|^2_2 \,  \Bigg[ \sum_{i=1}^N \Big( \int \di z  \xi^2_i(z) \Big)^2 + \sum_{\substack{i,j=1\\ i \neq j}}^N \Big( \int \di z  \xi_i(z) \xi_j(z) \Big)^2  \Bigg]^{1/2}\,.  \label{2.16}
\end{split} \ee
We proceed as in the proof of Prop.~\ref{prop:step1}. As for the diagonal term, using the scaling property $\mu(x) =N^{-1} \mu_i (y_i + x/N)$ we obtain
\be \begin{split} \label{Diag.Ci}
& \sum_{i=1}^N \Big( \int \di z\,  \xi^2_i(z) \Big)^2   \\ 
& \leq\sum_{i=1}^N \bigg[ \frac{C}{N^{2}}\int \di x \di {x^\prime}\m(x)\mu(x^\prime) |V(x)|^{1/2} |V(x^\prime)|^{1/2} \\
& \hskip 3cm  \times\Big( e^{- \sqrt \l \frac{ |x-x^\prime|}{N}} - e^{- \sqrt \l \frac{ |x|}{N}} - e^{- \sqrt \l \frac{|x^\prime|}{N}}+1 \Big)  \bigg]^2 \\
& \leq  C\sum_{i=1}^N \bigg[ N^{-3}\int \di x \di {x^\prime} \m(x)\m(x^\prime)|V(x)|^{1/2} |V(x^\prime)|^{1/2} (|x| +|x^\prime|) \bigg]^2 \\
&\leq C  N^{-6}  \| (1 +|\cdot|^2)V \|_1   \sum_{i=1}^N  \|\mu\|^4_2 \,    \leq C N^{-5}\,.
\end{split} \ee
Here we used $\| \mu \|_2^2 \leq C $, see \eqref{normaHatMu}.  To estimate the non diagonal term in \eqref{2.16} we use the bound \eqref{differenzaG} for the function $\zeta_{ij}^N$ defined in \eqref{zetaij}. By the scaling properties of $u_i(x)$ we have:
\begin{equation} \begin{split} \label{nonDiag.Ci}
& \sum_{\substack{i,j=1\\ i \neq j}}^N \Big( \int \di z\,  \xi_i(z) \xi_j(z) \Big)^2  \\
& = \sum_{\substack{i,j=1\\ i \neq j}}^N \bigg[ N^{-2} \int  \di x \di x^\prime |V(x)|^{1/2} |V(x^\prime)|^{1/2} \mu(x)\mu(x^\prime) \\
& \phantom{{}={}}\times \int \di z \left( \mathcal{G}^\l(y_i-z + x/ N) - \mathcal{G}^\l (y_i-z) \right) \left(\mathcal{G}^\l(y_j-z + x^\prime/N) - \mathcal{G}^\l(y_j-z)\right)  \bigg]^2 \\
& = \sum_{\substack{i,j=1\\ i \neq j}}^N \bigg[ N^{-2} \int  \di x \di x^\prime |V(x)|^{1/2} |V(x^\prime)|^{1/2} \m(x)\m(x^\prime)\, \zeta_{ij}^N(x,x^\prime) \Big]^2 \\
& \leq  C\,\sum_{\substack{i,j=1\\ i \neq j}}^N \bigg[ N^{-4} \int \di x \di x^\prime |V(x)|^{1/2} |V(x^\prime)|^{1/2} \m(x)\m(x^\prime) (|x| +|x^\prime|)^2   \bigg]^2  \\
& \leq C\, N^{-8}  \sum_{\substack{i,j=1\\ i \neq j}}^N  \Big( \int \di x |V(x)| (1 +|x|^2)  \int \di x^\prime |\m(x^\prime)|^2 \Big)^2  \leq C N^{-6}\,.
 \end{split}\end{equation}
Putting together \eqref{2.16}, \eqref{Diag.Ci} and \eqref{nonDiag.Ci} we obtain
 \be \label{normD}
 \|\vec D\, \|^2 \leq C N^{-5}\,.
 \ee
The bound \eqref{normR} for $\| \vec R \, \|$ follows from \eqref{normA}, \eqref{normB} and \eqref{normD}. 
\end{proof}

%
%

\newpage
\section{Proof of Theorem \ref{MainTheorem}}  \label{sec:convergence}
In this section we prove the main result stated in Theorem \ref{MainTheorem}. By Props.~\ref{prop:step1} and \ref{prop:step2} it remains to show the convergence of $\hat{\psi}_N$ to $\psi$. Although the proof is a slight modification of the step followed in~\cite{FHT} (see also \cite{BFT}) we report the details here for the sake of completeness. 


\begin{proposition}\label{prop:step3}
	Let $\hat{\psi}_N$  and $\psi$ be defined as in \eqref{eq:psi_N_hat} and \eqref{eq:psi} respectively. Then under the assumptions of Theorem~\ref{MainTheorem} and for $\lambda>\lambda_0$
	\[
		\lim_{N\to+\infty} N^{\beta}\lVert \hat{\psi}_N-\psi\rVert=0 \qquad \forall \beta<1/2.
	\]
\end{proposition}

To prove the proposition we first introduce $q$ defined by
\be\label{eq:q_def}
	-\frac{1}{4\pi a}q=\psi=(-\Delta+4\pi aW+\lambda)^{-1}f.
\ee
Then by the second resolvent identity we get
\be\label{eq:q}
	\frac{1}{4\pi a}q(x)+\int dz\, \mathcal{G}^\lambda(x-z)W(z)q(z)=-(\mathcal{G}^\lambda f)(x).
\ee
In the following lemma we compare $q(y_i)$ with $q_i.$
\begin{lemma}\label{lm:q_i-q}
	Let $q_i$ and $q$ be defined as in \eqref{eq:q_i} and \eqref{eq:q} respectively. Then under the same assumptions as in Theorem \ref{MainTheorem} and for $\lambda>\lambda_0$
	\[
		\lim_{N\to+\infty}N^{\beta}\left\{\frac{1}{N}\sum_{i=1}^N\,[N q_i-q(y_i)]^2\right\}^{1/2}=0 \qquad\forall\beta<1/2.
	\]
\end{lemma}
\begin{proof}
	From \eqref{eq:q_i} and \eqref{eq:q} we get
	\be\label{eq:q_i-q}
		\sum_{i=1}^N\left(\frac{1}{4\pi a}\delta_{ij}+\frac{1}{N}(1-\delta_{ij})G^\lambda_{ij}\right)\left(\sqrt{N}q_j-\frac{1}{\sqrt{N}}q(y_j)\right)=L_i
	\ee
	where
	\be\label{eq:L}
		L_i=\frac{1}{N^{3/2}}\sum_{\substack{j=1\\j\neq i}}^NG^\lambda_{ij}q(y_j)- \frac{1}{\sqrt{N}}(\mathcal{G}^\lambda Wq)(y_i).
	\ee
	Recalling the definition of $\G_{ij}^\lambda$ given in \eqref{eq:Gamma} we rewrite \eqref{eq:q_i-q} as
	\[
		\sum_{j=1}^N \G^\lambda_{ij} \left(\sqrt{N} q_j-\frac{1}{\sqrt{N}}q(y_j)\right)=(4\pi a) L_i.
	\]
	Using invertibility of $\Gamma^{\lambda}$ for $\lambda>\lambda_{0}$ (see \eqref{eq:lambda0}) and multiplying by $N^{\beta}$ with $\beta<1/2$ we get	\[
	N^{\beta}\left\{\sum_{i=1}^N	\left |\sqrt{N}q_i-\frac{1}{\sqrt{N}}q(y_i)\right |^2\right\}^{1/2} \leq C N^{\beta}\Vert \vec L\, \Vert\,.
	\]
		It remains to prove that $N^{\beta}\Vert \vec L\Vert$ goes to zero. In particular noticing that $E(\|\vec L\|)=0$ and applying Chebyshev inequality it is enough to show $N^{2\beta}E(\lVert \vec L \rVert^2)\to 0$. We use that
 \be \begin{split} \label{expectations}
&  E\big(\,\mathcal{G}^\l(x-y_i) q(y_i)\, \big) = \big(\mathcal{G}^\l W q \big)(x) \\
 &  E\big(\,(\mathcal{G}^\l W q)^2(y_i)\, \big)= \| \mathcal{G}^\l W q \|^2_{L^2_W} \\
 &  E\big(\,(\mathcal{G}^\l (y_i -y_j) q(y_j))^2\, \big) = \big(1, (\mathcal{G}^\l)^2 \ast (W q^2) \big)_{L^2_W} \,,
 \end{split}
 \ee
 where we used the notation $ (f\ast g)(x) = \int \di y f(x-y) g(y)$.  From \eqref{eq:L} and \eqref{expectations}  we obtain 
	\[ \begin{split}
			& N^{2\beta}E\left( \Vert \vec L \, \Vert^2\right) \\
		 & = N^{2\beta-1}E\Bigg(\sum_{i=1}^N \bigg(\frac{1}{N^2}\sum_{\substack{j,k=1\\j \neq i, k\neq i,j}}^N G^\lambda_{ij}G^\lambda_{ik}q(y_j)q(y_k) + \frac{1}{N^2}\sum_{\substack{j=1\\j\neq i}}^N \big(G^\lambda_{ij}q(y_j)\big)^2 \\
&   	\hskip 3.5cm	-\frac{2}{N}\sum_{\substack{j=1\\j\neq i}}^NG^\lambda_{ij}\,q(y_j)(\mathcal{G}^\lambda Wq)(y_i)\Big) +(\mathcal{G}^\lambda Wq )^2(y_i)\bigg) \;\Bigg)	\\
& = N^{2\beta-1}  \Bigg( \frac{(N-1)}{N} E\big( (\mathcal{G}^\l (y_1 -y_2) q(y_2))^2 \big) \\
& \hskip 2cm + \bigg( \frac{(N-1)(N-2)}{N} - 2(N-1) +N \bigg)  E\big((\mathcal{G}^\l W q)^2(y_1)\big)\Bigg)	\\
 &   = \frac{N-1}{N^{2-2\beta}}\,(1,(\mathcal{G}^{\lambda})^2 \ast (Wq^2))_{L^2_W}-\frac{N-2}{N^{2-2\beta}}\,\Vert \mathcal{G}^\lambda Wq\Vert^2_{L^2_W}\,,
		\end{split}
	\]
	which goes to zero for $N \to \infty$ for any $\b<1/2$.
This concludes the proof of the Lemma.
\end{proof}

\begin{proof} [Proof of Prop.~\ref{prop:step3}]
		Let us consider $g\in L^2(\R^3).$ Then by \eqref{eq:q_i}, \eqref{eq:psi_N_hat}, \eqref{eq:q_def},\eqref{eq:q} we get
		\[
			\begin{aligned}
				|(g,\hat{\psi}_N-\psi)|&=|(g,\sum_{i=1}^N q_i\,\mathcal{G}^\lambda (\cdot-y_i))-(g,\mathcal{G}^\lambda Wq)
				|\\
											      &\leq |\sum_{i=1}^N \big(q_i-\frac{q(y_i)}{N}\big)\mathcal{G}^\lambda g(y_i)|+|\frac{1}{N}\sum_{i=1}^N q(y_i) \mathcal{G}^\lambda g(y_i)-(g,\mathcal{G}^\lambda Wq)|.
			\end{aligned}
		\]
		Then using Cauchy-Schwarz inequality and multiplying both sides by $\displaystyle{\frac{N^{\beta}}{\|g\|}}$ we obtain
		\be\label{eq:psi_hat-psi}
			\begin{aligned}
			\frac{|N^{\beta}(g,\hat{\psi}_N-\psi)|}{\| g \|}\leq &\frac{\displaystyle{\sup_x} |\mathcal{G}^\lambda g(x)|}{\| g \|}N^{\beta}\left\{\frac{1}{N}\sum_{i=1}^N(Nq_i-q(y_i))^2\right\}^{1/2}\\
																			&+N^{\beta}\frac{\left|\eta(Y_N)-E(\eta(Y_N))\right|}{\lVert g \rVert}
		\end{aligned}
		\ee
	where 
	\[
		\eta(Y^N)=\frac{1}{N}\sum_{i=1}^N(\mathcal{G}^\lambda g)(y_i)q(y_i).
	\]
	The first term in \eqref{eq:psi_hat-psi} goes to zero by Lemma \ref{lm:q_i-q}. Furthermore
	\[
		\begin{aligned}
			E\!\left(\frac{|\eta(Y_N)-E(\eta(Y_N))^2	}{\lVert g \rVert^2}\right)\!\!=&\frac{E\left(\eta(Y_N)^2\right)}{\| g \|^2}-\frac{E\left(\eta(Y_N)\right)^2}{\lVert g \rVert^2}\\
											     =&\left(\frac{\int\! \di x\, W(x)q^2(x)(\mathcal{G}^\lambda g)^2(x)}{\|g\|^{2}N}+\frac{N-1}{N}\frac{E(\eta(Y_N))^2}{\|g\|^{2}}\right)\\
											        &-\frac{E(\eta(Y_N))^2}{\lVert g\rVert^2}\\
											     =&\frac{\int\! \di x\, W(x)q^2(x)(\mathcal{G}^\lambda g)^2(x)\!-\!\left(\int\! \di y\, (\mathcal{G}^\lambda g)(y) q(y) W(y)\right)^2}{\|g\|^{2}\,N}\\
											    \leq& \frac{C}{N}(\lVert q\rVert_{L^2_W}+(1,q)_{L^2_W}^2).   
		\end{aligned}
	\]
	Then by Chebyshev inequality also the second term in \eqref{eq:psi_hat-psi} goes to zero uniformily in $\lVert g\rVert$. Taking the supremum over $g\in L^2(\bR^3) $ we get the thesis.
\end{proof}
We are now ready to prove our main result
\begin{proof}[Proof of Theorem \ref{MainTheorem}]
	It follows immediately from Propositions \ref{prop:step1}, \ref{prop:step2} and \ref{prop:step3}.
\end{proof}

\appendix

%
%



\end{document}